\documentclass[11pt]{article}
\usepackage[left=1in, right=1in, top=1in, bottom=1in]{geometry}
\usepackage{CJK}
\usepackage{latexsym,bm}
\usepackage{xypic}
\usepackage{amsmath}
\usepackage{wasysym}
\usepackage{indentfirst}
\usepackage{amssymb}
\usepackage{dsfont}
\usepackage{amsthm}
\begin{document}
\newcommand{\fr}[2]{\frac{\;#1\;}{\;#2\;}}
\newtheorem{theorem}{Theorem}[section]
\newtheorem{lemma}{Lemma}[section]
\newtheorem{proposition}{Proposition}[section]
\newtheorem{corollary}{Corollary}[section]
\newtheorem{conjecture}{Conjecture}[section]
\newtheorem{remark}{Remark}[section]
\newtheorem{definition}{Definition}[section]
\newtheorem{example}{Example}[section]
\newtheorem{notation}{Notation}[section]
\numberwithin{equation}{section}
\newcommand{\Aut}{\mathrm{Aut}\,}
\newcommand{\CSupp}{\mathrm{CSupp}\,}
\newcommand{\Supp}{\mathrm{Supp}\,}
\newcommand{\rank}{\mathrm{rank}\,}
\newcommand{\col}{\mathrm{col}\,}
\newcommand{\len}{\mathrm{len}\,}
\newcommand{\leftlen}{\mathrm{leftlen}\,}
\newcommand{\rightlen}{\mathrm{rightlen}\,}
\newcommand{\length}{\mathrm{length}\,}
\newcommand{\bin}{\mathrm{bin}\,}
\newcommand{\wt}{\mathrm{wt}\,}
\newcommand{\Wt}{\mathrm{Wt}\,}
\newcommand{\diff}{\mathrm{diff}\,}
\newcommand{\lcm}{\mathrm{lcm}\,}
\newcommand{\GL}{\mathrm{GL}\,}
\newcommand{\SJ}{\mathrm{SJ}\,}
\newcommand{\LG}{\mathrm{LG}\,}
\newcommand{\bij}{\mathrm{bij}\,}
\newcommand{\dom}{\mathrm{dom}\,}
\newcommand{\fun}{\mathrm{fun}\,}
\newcommand{\SUPP}{\mathrm{SUPP}\,}
\newcommand{\supp}{\mathrm{supp}\,}
\newcommand{\End}{\mathrm{End}\,}
\newcommand{\Hom}{\mathrm{Hom}\,}
\newcommand{\ran}{\mathrm{ran}\,}
\newcommand{\row}{\mathrm{row}\,}
\newcommand{\Mat}{\mathrm{Mat}\,}
\newcommand{\rk}{\mathrm{rk}\,}
\newcommand{\rs}{\mathrm{rs}\,}
\newcommand{\piv}{\mathrm{piv}\,}
\newcommand{\Tr}{\mathrm{Tr}\,}
\newcommand{\perm}{\mathrm{perm}\,}
\newcommand{\rsupp}{\mathrm{rsupp}\,}
\newcommand{\inv}{\mathrm{inv}\,}
\newcommand{\orb}{\mathrm{orb}\,}
\newcommand{\id}{\mathrm{id}\,}
\newcommand{\soc}{\mathrm{soc}\,}
\newcommand{\unit}{\mathrm{unit}\,}
\newcommand{\word}{\mathrm{word}\,}
\newcommand{\Sym}{\mathrm{Sym}\,}
\newcommand{\Jac}{\mathrm{Jac}\,}

\renewcommand{\thefootnote}{\fnsymbol{footnote}}

\title{Reflexive Partitions Induced by Rank Support and Non-Reflexive Partitions Induced by Rank Weight}
\author{Yang Xu$^1$ \,\,\,\,\,\, Haibin Kan$^2$\,\,\,\,\,\,Guangyue Han$^3$}

\maketitle

\renewcommand{\thefootnote}{\fnsymbol{footnote}}

\footnotetext{\hspace*{-6mm} \begin{tabular}{@{}r@{}p{16cm}@{}}
$^1$ & Shanghai Key Laboratory of Intelligent Information Processing, School of Computer Science, Fudan University,
Shanghai 200433, China.\\
&Shanghai Engineering Research Center of Blockchain, Shanghai 200433, China. {E-mail:xuyyang@fudan.edu.cn}\\
$^2$ & Shanghai Key Laboratory of Intelligent Information Processing, School of Computer Science, Fudan University,
Shanghai 200433, China.\\
&Shanghai Engineering Research Center of Blockchain, Shanghai 200433, China.\\
&Yiwu Research Institute of Fudan University, Yiwu City, Zhejiang 322000, China. {E-mail:hbkan@fudan.edu.cn} \\
$^3$ & Department of Mathematics, Faculty of Science, The University of Hong Kong, Pokfulam Road, Hong Kong, China. {E-mail:ghan@hku.hk} \\
\end{tabular}}

\vskip 3mm

{\hspace*{-6mm}{\bf Abstract}---In this paper, we study partitions of finite modules induced by rank support and rank weight. First, we show that partitions induced by rank support are mutually dual with respect to suitable non-degenerate pairings, and hence are reflexive; moreover, we compute the associated generalized Krawtchouk matrices. Similar results are established for partitions induced by isomorphic relation of rank support. These results generalize counterpart results established for row space partitions and rank partitions of matrix spaces over finite fields. Next, we show that partitions of free modules over a finite chain ring $R$ induced by rank weight are non-reflexive provided that $R$ is not a field; moreover, we characterize the dual partitions explicitly. As a corollary, we show that rank partitions of matrix spaces over $R$ are reflexive if and only if $R$ is a field; moreover, two matrices belong to the same member of the dual partition if and only if their transposes are equivalent. In particular, we show that opposite to matrices over finite fields, rank metric does not induce an association scheme provided that $R$ is not a field, which further settles an open question proposed by Blanco-Chac\'{o}n, Boix, Greferath and Hieta-Aho in \cite{2}.
}\!

\section{Introduction}
A rank metric code is a set of matrices over a finite field endowed with rank metric, i.e., the distance between two matrices is the rank of their difference (see \cite{10,11}). Rank metric codes were first proposed and studied by Delsarte in \cite{6} via an association scheme approach. One of the most important results established in \cite{6} is the MacWilliams identity, where it was shown that rank metric induces an association scheme, which further allows the derivation of MacWilliams identities which relate the rank weight distributions of an additive code with those of its dual code.

Several other approaches towards MacWilliams identities for rank metric codes have been proposed in recent years. In \cite{17}, Ravagnani established an equivalent version of the counterpart identities in \cite{6} by using elementary linear algebra. Later in \cite{11}, Gorla and Ravagnani re-established both the identities in \cite{6,16} via an elementary approach. In \cite{8}, Gluesing-Luerssen recovered the identities in \cite{6} from general MacWilliams identities for partitions of finite abelian groups by observing that the partition induced by rank weight is an orbit partition, which forms a special instance of reflexive partitions. In \cite{18}, Ravagnani recovered the identities in \cite{6} from general MacWilliams identities for regular supports defined from finite abelian groups to regular lattices.

More recently, rank metric codes defined over finite principal ideal rings were proposed and studied by Kamche and Mouaha in \cite{14}, and by Chac\'{o}n, Boix, Greferath and Hieta-Aho in \cite{2}. In particular, in \cite{2}, the authors extend Ravagnani's MacWilliams duality theory in \cite{17} to rank metric codes over finite chain rings by establishing MacWilliams identities relating $q$-binomial moments of a free rank metric code with those of its dual code. Moreover, an interesting open question was proposed in \cite{2}, namely, whether rank metric defined for matrices over finite chain rings still forms an association scheme (just like rank metric defined for matrices over finite fields), and if so, whether it is possible to compute its parameters.

In this paper, we settle the above mentioned open question by showing that somehow surprisingly, opposite to rank metric defined over finite fields, rank metric defined over finite chain rings does not form an association scheme as long as the ring is not a field, and the matrices have at least $2$ rows and $2$ columns. Our approach is based on the theory of partitions and dual partitions proposed in \cite{8,9}. More specifically, we will study three classes of partitions induced by rank support and rank weight, and give explicit characterizations of their dual partitions; moreover, the non-reflexivity of the rank partition will leads us to the desired result.

The rest of this paper is organized as follows.

In Section 2, we give some preliminaries including partitions of finite abelian groups and MacWilliams identities (Section 2.1), basic properties of bilinear maps defined for modules (Section 2.2) and basic properties of modules over finite chain rings (Section 2.3). In Section 3, we study partitions of modules induced by rank support and induced by isomorphic relation of rank support, which are generalizations of row space partitions and rank partitions of matrix spaces over finite fields (see \cite{9}). We first compute the associated generalized Krawtchouk matrices by using the M\"{o}bius function of the submodule lattice (Proposition 3.1), and then show that both these partitions are mutually dual with respect to suitable non-degenerate pairings, and hence are reflexive (Theorems 3.1 and 3.2). Our results can be readily applied to modules over finite Frobenius rings, which further recover counterpart results established for matrix spaces over finite fields in \cite{9}.

In Section 4, we study partitions of free modules over finite chain rings induced by rank weight, where the rank of the module is at least $2$. We show that the dual partitions of such partitions are exactly the partitions induced by isomorphic relation of rank support, and consequently, partitions induced by rank weight are non-reflexive provided that the chain ring is not a field (Theorem 4.1). More explicit necessary and sufficient conditions for two elements to belong to the same member of the dual partition are also provided (Theorem 4.2).

In Section 5, we study the rank partitions of matrix spaces over finite chain rings, where the matrices have at least $2$ rows and $2$ columns. With the help of results established in Section 4, we characterize the dual partition of the rank partition by showing that among others, two matrices belong to the same member of the dual partition if and only if their transposes are equivalent; moreover, the rank partition is reflexive if and only if the chain ring is a field (Theorem 5.1). Finally, using the connection between reflexive partitions and association schemes, we show that rank metric induces an association scheme if and only if the chain ring is a field (Theorem 5.2), which further settles the aforementioned question proposed in \cite{2}.

\section{Preliminaries}
We begin with a few notations and terminologies that will be used throughout the rest of the paper. First of all, we let $\mathbb{C}^{\ast}=\mathbb{C}-\{0\}$ denote the multiplicative group of the complex numbers. For any $a,b\in\mathbb{Z}$, we use $[a,b]$ to denote the set of all the integers
between $a$ and $b$, i.e., $[a,b]=\{i\in\mathbb{Z}\mid a\leqslant i\leqslant b\}$.

Next, consider a finite set $E$. A \textit{partition} of $E$ is a collection of nonempty and disjoint subsets of $E$ whose union is $E$. Consider a partition $\Gamma$ of $E$. For any $u,v\in E$, we write $u\sim_{\Gamma}v$ if $u$ and $v$ belong to the same member of $E$. For any $D\subseteq E$, the \textit{$\Gamma$-distribution} of $D$ is defined as the sequence $(|D\cap B|\mid B\in\Gamma)$. For two partitions $\Gamma,\Psi$ of $E$, we say that \textit{$\Gamma$ is finer than $\Psi$}, if for any $u,v\in E$, $u\sim_{\Gamma}v$ implies that $u\sim_{\Psi}v$.

Now consider a ring $R$, which is always assumed to be associative with a multiplicative identity. For any $m,n\in\mathbb{Z}^{+}$, let $R^{[m]}$ denote the set of all the column vectors over $R$ of length $m$, and let $\Mat_{m,n}(R)$ denote the set of all the matrices over $R$ with $m$ rows and $n$ columns. The \textit{standard inner product} on $R^{[m]}$ is the map from $R^{[m]}\times R^{[m]}$ to $R$ defined as
$$(a,b)\mapsto\sum_{i=1}^{m}a_ib_i,$$
where $a_i$ denotes the $i$-th entry of $a\in R^{[m]}$; moreover, following \cite{11}, the \textit{trace inner product} on $\Mat_{m,n}(R)$ is the map from $\Mat_{m,n}(R)\times \Mat_{m,n}(R)$ to $R$ defined as
$$(\alpha,\beta)\mapsto \Tr(\alpha^{T}\beta),$$
where $\alpha^{T}$ denotes the transpose of $\alpha\in\Mat_{m,n}(R)$, and $\Tr$ denotes the trace of matrices. For a finite left (or right) $R$-module $P$, let $\rank_R(P)$ denote the minimal cardinality of a generator set of $P$.

\subsection{Partitions of finite abelian groups and MacWilliams identities}
Let $G$ and $H$ be finite abelian groups, and let $f:G\times H\longrightarrow \mathbb{C}^{*}$ be a \textit{pairing}, i.e., for any $a,c\in G$ and $b,d\in H$, it holds that $f(ac,b)=f(a,b)f(c,b)$, $f(a,bd)=f(a,b)f(a,d)$ (see [16, Definition 11.7]). For any $C\subseteq G$ and $D\subseteq H$, define $C^{\ddagger}\leqslant H$ and $^{\ddagger}D\leqslant G$ as
\begin{equation}C^{\ddagger}\triangleq\{b\in H\mid \forall~a\in C:f(a,b)=1\},\end{equation}
\begin{equation}{^{\ddagger}D}\triangleq\{a\in G\mid \forall~b\in D:f(a,b)=1\}.\end{equation}
We further assume that $f$ is \textit{non-degenerate}, i.e., $G^{\ddagger}=\{1_{H}\}$, ${^{\ddagger}H}=\{1_{G}\}$ (see [16, Definition 11.7]). Note that the non-degenerate condition implies that $G\cong H$ as groups, and conversely, $G\cong H$ implies the existence of such a non-degenerate pairing (see [16, Lemma 11.8]).

For a partition $\Gamma$ of $H$, \textit{the left dual partition of $\Gamma$ with respect to $f$}, denoted by $\textbf{\textit{l}}(\Gamma)$, is the partition of $G$ such that for any $a,c\in G$, $a\sim_{\textbf{\textit{l}}(\Gamma)}c$ if and only if
\begin{equation}\forall~B\in\Gamma:\sum_{b\in B}f(a,b)=\sum_{b\in B}f(c,b).\end{equation}
Similarly, for a partition $\Lambda$ of $G$, \textit{the right dual partition of $\Lambda$ with respect to $f$}, denoted by $\textbf{\textit{r}}(\Lambda)$, is the partition of $H$ such that for any $b,d\in H$, $b\sim_{\textbf{\textit{r}}(\Lambda)}d$ if and only if
\begin{equation}\forall~A\in\Lambda:\sum_{a\in A}f(a,b)=\sum_{a\in A}f(a,d).\end{equation}
The notion of partition and dual partition provides a general framework for addressing MacWilliams identities. More precisely, let $\Gamma$ be a partition of $H$, and let $\Lambda$ be a partition of $G$. When $\Lambda$ is finer than $\textbf{\textit{l}}(\Gamma)$, the \textit{left generalized Krawtchouk matrix of $(\Lambda,\Gamma)$ with respect to $f$} is defined as $\rho:\Lambda\times\Gamma\longrightarrow \mathbb{C}$, where for any $(A,B)\in\Lambda\times\Gamma$,
\begin{equation}\text{$\rho(A,B)=\sum_{b\in B}f(a,b)$ for any chosen $a\in A$.}\end{equation}
In this case, it has been proven in [8, Theorem 2.7] that for an additive code (i.e., subgroup) $C\leqslant G$, the $\Lambda$-distribution of $C$ determines the $\Gamma$-distribution of $C^{\ddagger}$ via the following MacWilliams identity
\begin{equation}\forall~B\in\Gamma:|C||{C^{\ddagger}}\cap B|=\sum_{A\in \Lambda}|C\cap A|\cdot\rho(A,B).\end{equation}
Similarly, when $\Gamma$ is finer than $\textbf{\textit{r}}(\Lambda)$, the \textit{right generalized Krawtchouk matrix of $(\Lambda,\Gamma)$ with respect to $f$} is defined as $\varepsilon:\Lambda\times\Gamma\longrightarrow \mathbb{C}$, where for any $(A,B)\in\Lambda\times\Gamma$, \begin{equation}\text{$\varepsilon(A,B)=\sum_{c\in A}f(c,d)$ for any chosen $d\in B$.}\end{equation}
In this case, for an additive code $D\leqslant H$, the $\Gamma$-distribution of $D$ determines the $\Lambda$-distribution of ${^{\ddagger}D}$ via the following MacWilliams identity
\begin{equation}\forall~A\in\Lambda:|D||{^{\ddagger}D}\cap A|=\sum_{B\in\Gamma}|D\cap B|\cdot\varepsilon(A,B).\end{equation}

The following definition follows [8, Definition 2.1] and [22, Definition 2].

\setlength{\parindent}{0em}
\begin{definition}
{Let $\Gamma$ be a partition of $H$, and let $\Lambda$ be a partition of $G$. If $\Lambda$ is finer than $\textbf{\textit{l}}(\Gamma)$ and $\Gamma$ is finer than $\textbf{\textit{r}}(\Lambda)$, then we say that $(\Lambda,\Gamma)$ is mutually dual with respect to $f$. Moreover, $\Gamma$ is said to be reflexive if $\Gamma=\textbf{\textit{r}}(\textbf{\textit{l}}(\Gamma))$.
}
\end{definition}

\setlength{\parindent}{2em}
Now we collect some important properties of reflexive partitions. The following lemma is a consequence of [8, Theorem 2.4].

\setlength{\parindent}{0em}
\begin{lemma}
{\bf{(1)}}\,\,Let $\Gamma$ be a partition of $H$. Then, we have $\{1_{G}\}\in\textbf{\textit{l}}(\Gamma)$, $|\Gamma|\leqslant|\textbf{\textit{l}}(\Gamma)|$, $\textbf{\textit{r}}(\textbf{\textit{l}}(\Gamma))$ is finer than $\Gamma$. Moreover, $\Gamma$ is reflexive if and only if  $|\Gamma|=|\textbf{\textit{l}}(\Gamma)|$.

{\bf{(2)}}\,\,Let $\Gamma$ be a partition of $H$, and let $\Lambda$ be a partition of $G$ such that $(\Lambda,\Gamma)$ is mutually dual with respect to $f$. Then, we have $|\Lambda|=|\Gamma|$, $\Lambda=\textbf{\textit{l}}(\Gamma)$, $\Gamma=\textbf{\textit{r}}(\Lambda)$, and both $\Gamma$ and $\Lambda$ are reflexive.
\end{lemma}

\setlength{\parindent}{2em}
Mutually dual reflexive partitions form an appealing case for MacWilliams identities. Indeed, MacWilliams identities derived from a reflexive partition is invertible, whose inverse is essentially the  MacWilliams identities derived from the dual partition. More precisely, following (2) of Lemma 2.1, and let $\rho:\Lambda\times\Gamma\longrightarrow \mathbb{C}$ and $\varepsilon:\Lambda\times\Gamma\longrightarrow \mathbb{C}$ be the left and right generalized Krawtchouk matrices of $(\Lambda,\Gamma)$ with respect to $f$, respectively. Then, for any $(U,V)\in\Gamma\times\Gamma$, it holds that
\begin{eqnarray*}
\begin{split}
\sum_{A\in\Lambda}\varepsilon(A,U)\cdot\overline{\rho(A,V)}=\begin{cases}
|G|,&U=V;\\
0,&U\neq V,
\end{cases}
\end{split}
\end{eqnarray*}
and for any $(I,J)\in\Lambda\times\Lambda$, it holds that
\begin{eqnarray*}
\begin{split}
\sum_{B\in\Gamma}\overline{\rho(I,B)}\cdot\varepsilon(J,B)=\begin{cases}
|G|,&I=J;\\
0,&I\neq J.
\end{cases}
\end{split}
\end{eqnarray*}

\setlength{\parindent}{2em}
In the following example, we consider reflexive partitions induced by Hamming support and Hamming weight.

\setlength{\parindent}{2em}
\begin{example}
Let $B$ be a finite abelian group with $|B|=q$, and let $\Omega$ be a non-empty finite set with $|\Omega|=n$. For any $\alpha\in B^{\Omega}$, the Hamming support of $\alpha$ is defined as $\supp(\alpha)\triangleq\{i\in\Omega\mid \alpha_i\neq1_B\}$, and the Hamming weight of $\alpha$ is defined as $\wt(\alpha)\triangleq|\supp(\alpha)|$. Let $\Delta_1$ and $\Delta_2$ denote the partitions of $B^{\Omega}$ induced by Hamming support and Hamming weight, respectively. Moreover, let $\varpi:B\times B\longrightarrow \mathbb{C}^{*}$ be a non-degenerate pairing, and define the non-degenerate pairing $f:B^{\Omega}\times B^{\Omega}\longrightarrow \mathbb{C}^{*}$ as
$$f(\alpha,\gamma)=\prod_{i\in\Omega}\varpi(\alpha_i,\gamma_i).$$
Consider $\alpha\in B^{\Omega}$ with $D=\supp(\alpha)$ and $s=\wt(\alpha)$. Then, for any $I\subseteq\Omega$, we have
$$\sum_{(\beta\in B^{\Omega},\supp(\beta)=I)}f(\alpha,\beta)=(-1)^{|I\cap D|}(q-1)^{|I-D|},$$
and for any $m\in[0,n]$, we have
$$\sum_{(\beta\in B^{\Omega},\wt(\beta)=m)}f(\alpha,\beta)=\sum_{t=0}^{m}(-1)^{t}(q-1)^{m-t}\binom{s}{t}\binom{n-s}{m-t}.$$
Consequently, both $(\Delta_1,\Delta_1)$ and $(\Delta_2,\Delta_2)$ are mutually dual with respect to $f$, and hence both $\Delta_1$ and $\Delta_2$ are reflexive (see [8, Example 2.3 (c)] and [18, Example 36]).
\end{example}

\setlength{\parindent}{2em}
We end this subsection by noting that reflexivity of partitions is independent of the choice of the non-degenerate pairing. In fact, reflexive partitions can be characterized in terms of association schemes (see \cite{5,7,22}), as detailed in the following remark.

\setlength{\parindent}{0em}
\begin{remark}
Let $\Gamma$ be a partition of $H$. As a corollary of [22, Theorem 1] (also see the
discussion after [7, Theorem 2.4]), $\Gamma$ is reflexive if and only if $\{1_H\}\in\Gamma$, and for any $U,V,W\in\Gamma$ and $w,z\in W$, it holds that
$$|\{(u,v)\in U\times V\mid uv=w\}|=|\{(u,v)\in U\times V\mid uv=z\}|.$$
Alternatively speaking, let $\Theta$ be the partition of $H^{2}$ defined as $(u,v)\sim_{\Theta}(x,y)\Longleftrightarrow u^{-1}v\sim_{\Gamma}x^{-1}y$. Then, $\Gamma$ is reflexive if and only if $(H,\Theta)$ is an association scheme.
\end{remark}

\subsection{Bilinear maps defined for modules}

\setlength{\parindent}{2em}
Throughout this subsection, let $R$ and $S$ be rings, $P$ be a finite left $R$-module, $Q$ be a finite right $S$-module, $U$ be an $(R,S)$-bimodule, and let $g:P\times Q\longrightarrow U$ be an $(R,S)$-bilinear map, i.e., for any $a,c\in P$, $b,d\in Q$, $r\in R$ and $s\in S$, it holds that
$$g(a+c,b)=g(a,b)+g(c,b),$$
$$g(a,b+d)=g(a,b)+g(a,d),$$
$$g(ra,bs)=rg(a,b)s.$$
For any $A\subseteq P$ and $B\subseteq Q$, define
\begin{equation}A^{\bot}=\{v\in Q\mid\forall~u\in A:g(u,v)=0\},\end{equation}
\begin{equation}{^{\bot}B}=\{u\in P\mid \forall~v\in B:g(u,v)=0\}.\end{equation}
We further assume that $g$ is \textit{non-degenerate}, i.e., $P^{\bot}=\{0\}$, ${^{\bot}Q}=\{0\}$ (see [1, Chapter 6, Section 24]).

In the following lemma, we use $g$ to obtain non-degenerate pairings.

\setlength{\parindent}{0em}
\begin{lemma}
Let $\chi\in\Hom((U,+),\mathbb{C}^{\ast})$ such that the only left $R$-submodule or right $S$-submodule of $U$ contained in $\ker(\chi)$ is $\{0\}$. Define $f:P\times Q\longrightarrow \mathbb{C}^{\ast}$ as $f(a,b)=\chi(g(a,b))$, i.e., $f=\chi\circ g$. Then, $f$ is a non-degenerate pairing. Moreover, for any $A\leqslant_{R} P$ and $B\leqslant_{S} Q$, it holds that
\begin{equation}A^{\bot}=\{v\in Q\mid\forall~u\in A:f(u,v)=1\},\end{equation}
\begin{equation}{^{\bot}B}=\{u\in P\mid \forall~v\in B:f(u,v)=1\}.\end{equation}
\end{lemma}

\begin{proof}
First, we prove (2.11). Consider $v\in Q$. If $v\in A^{\bot}$, then for any $u\in A$, we have $g(u,v)=0$, which implies that $f(u,v)=\chi(0)=1$, as desired. Now suppose that $f(u,v)=1$ for all $u\in A$. Let $L=\{g(u,v)\mid u\in A\}$. Then, we have $L\subseteq\ker(\chi)$; moreover, it follows from $A\leqslant_{R} P$ and $g$ is $R$-linear that $L$ is a left $R$-submodule of $U$. Therefore we have $L=\{0\}$, which implies that $v\in A^{\bot}$, as desired. Second, (2.12) can be proved in a parallel way. Finally, noticing that $P^{\bot}=\{0\}$, ${^{\bot}Q}=\{0\}$,  (2.11) and (2.12) immediately imply that $f$ is non-degenerate, as desired.
\end{proof}

\setlength{\parindent}{2em}
The assumption for $\chi$ in Lemma 2.2 is closely related to the notion of \textit{finite Frobenius rings} (see \cite{8,13,20}). More precisely,
when $R$ is finite, a character $\chi\in\Hom((R,+),\mathbb{C}^{\ast})$ is referred to as a \textit{generating character of $R$} if the only left ideal of $R$ contained in $\ker(\chi)$ is $\{0\}$ (or equivalently, the only right ideal of $R$ contained in $\ker(\chi)$ is $\{0\}$, see [20, Theorem 4.3]). It is well-known that a finite ring has a generating character if and only if it is a Frobenius ring, i.e., its left socle is a principle left ideal (see [13, Theorems 1 and 2]). Hence if $R=S=U$ are finite, then $\chi$ satisfies the assumption of Lemma 2.2 if and only if $R$ is a Frobenius ring and $\chi$ is a generating character of $R$.

Lemma 2.2 can be used to derive MacWilliams identities. More precisely, following Lemma 2.2 and Section 2.1, let $\Gamma$ be a partition of $Q$, and let $\Lambda$ be a partition of $P$. When $\Lambda$ is finer than $\textbf{\textit{l}}(\Gamma)$, then for any $C\leqslant_{R}P$, the $\Lambda$-distribution of $C$ determines the $\Gamma$-distribution of $C^{\bot}$ via the MacWilliams identity (2.6); similarly, when $\Gamma$ is finer than $\textbf{\textit{r}}(\Lambda)$, then for any $D\leqslant_{S}Q$, the $\Gamma$-distribution of $D$ determines the $\Lambda$-distribution of ${^{\bot}D}$ via the MacWilliams identity (2.8).

\setlength{\parindent}{2em}
Next, we recall the notion of \textit{pseudo-injective modules}. Following [21, Section 5.3], the left $R$-module $P$ is said to be pseudo-injective if for any $A\leqslant_{R} P$ and any injective $h\in\Hom_{R}(A,P)$, there exists $\varphi\in\End_{R}(P)$ with $\varphi\mid_{A}=h$. Since $P$ is finite, by [21, Proposition 5.1], if $P$ is pseudo-injective, then for any $A\leqslant_{R} P$, any injective $h\in\Hom_{R}(A,P)$ extends to a left $R$-module automorphism of $P$. Note that pseudo-injectivity can be defined in a parallel way for right $S$-modules.

The following Lemma is a corollary of [1, Theorem 30.1], and will be used in Sections 3 and 4.

\setlength{\parindent}{0em}
\begin{lemma}
Suppose that for any simple left $R$-module $X$, $\Hom_{R}(X,U)$ is a simple right $S$-module, and for any simple right $S$-modules $Y$, $\Hom_{S}(Y,U)$ is a simple left $R$-module. Let $A\leqslant_{R} P$, $C\leqslant_{R} P$. Then, the following two statements hold true:

{\bf{(1)}}\,\,If $P$ is pseudo-injective and $A\cong C$ as left $R$-modules, then $A^{\bot}\cong C^{\bot}$ as right $S$-modules;

{\bf{(2)}}\,\,If $Q$ is pseudo-injective and $A^{\bot}\cong C^{\bot}$ as right $S$-modules, then $A\cong C$ as left $R$-modules.
\end{lemma}

\begin{proof}
We only prove (1) and the proof of (2) is similar. Since $P$ is pseudo-injective and $A\cong C$, we can choose $\varphi\in\Aut_{R}(P)$ with $\varphi[A]=C$. Moreover, $\psi:P/A\longrightarrow P/C$ defined as $\psi(x+A)=\varphi(x)+C$ for all $x\in P$ is a well-defined left $R$-module isomorphism. It follows that $\Hom_{R}(P/A,U)\cong\Hom_{R}(P/C,U)$ as right $S$-modules. Moreover, by [1, Theorem 30.1], we have $\Hom_{R}(P/A,U)\cong A^{\bot}$ and $\Hom_{R}(P/C,U)\cong C^{\bot}$, which further establishes the desired result.
\end{proof}

\subsection{Modules over finite chain rings}
\setlength{\parindent}{2em}
In this subsection, we follow [3, Section 2.1] and [12, Section 2] to collect some basic properties of modules over finite chain rings which we will use frequently in Section 4.

Throughout this subsection, let $R$ be a \textit{finite chain ring}, i.e., $R$ is a finite ring that has a unique maximal left ideal, denoted by $J$; moreover, $J$ is a principal left ideal. We fix $\pi\in R$ such that $J=R\pi$. In fact, $J$ is an ideal of $R$, and the quotient ring $R/J$ is a finite field, whose cardinality is denoted by $q$. There uniquely exists $s\in\mathbb{Z}^{+}$ such that $\pi^{s}=0$ and $\pi^{s-1}\neq0$. Then, for any $i\in[0,s]$, we have $R\pi^{i}=\pi^{i} R$, $|R/R\pi^{i}|=q^{i}$, and $R/R\pi^{i}\cong R\pi^{s-i}$ as left $R$-modules. We note that every left or right ideal of $R$ is equal to $R\pi^{i}$ for some $i\in[0,s]$; consequently, the set of all the ideals of $R$ forms a chain under the inclusion relation. We also note that $R$ is a Frobenius ring. In fact, a character $\chi\in\Hom((R,+),\mathbb{C}^{\ast})$ is a generating character of $R$ if and only if $R\pi^{s-1}\nsubseteq\ker(\chi)$.

Now we collect some facts on finite left $R$-modules in the following lemma, whose proof is given in the appendix. Similar results also hold true for finite right $R$-modules in a parallel fashion.

\setlength{\parindent}{0em}
\begin{lemma}
{\bf{(1)}}\,\,Let $M$ be a finite left $R$-module. Then, there exists $m\in\mathbb{Z}^{+}$ and $(\lambda_1,\dots,\lambda_m)\in[0,s]^{m}$ such that
\begin{equation}M\cong\prod_{i=1}^{m}R/R\pi^{\lambda_i}\cong\prod_{i=1}^{m}R\pi^{s-\lambda_i}\end{equation}
as left modules. Moreover, for any $t\in[0,s]$, it holds that
$$|\{y\in M\mid \pi^{t}y=0\}|=q^{\sum_{i=1}^{m}\min\{\lambda_i,t\}}.$$

{\bf{(2)}}\,\,Let $M$ and $P$ be two finite left $R$-modules such that $|\{y\in M\mid \pi^{t}y=0\}|=|\{z\in P\mid \pi^{t}z=0\}|$ for all $t\in[1,s]$. Then, $M\cong P$ as left $R$-modules.

{\bf{(3)}}\,\,Let $M$ be a free finite left $R$-module with $\rank_R(M)=m\geqslant1$, $A\leqslant_{R}M$, and let $B=\{y\in M \mid \pi y\in A\}$. Then, for any $t\in[1,s]$, it holds that
$$|\{z\in B\mid \pi^{t}z=0\}|=|\{y\in A\mid \pi^{t-1}y=0\}|\cdot q^{m}.$$

{\bf{(4)}}\,\,Let $M$ be a finite left $R$-module. Then, for any $t\in[1,s]$, it holds that
$$|\{A\leqslant_{R}M\mid\rank_R(A)=1,|A|=q^{t}\}|=\frac{|\{y\in M\mid \pi^{t}y=0\}|-|\{y\in M\mid \pi^{t-1}y=0\}|}{q^{t}-q^{t-1}}.$$
\end{lemma}

\section{Reflexive partitions induced by rank support}

\setlength{\parindent}{2em}
Throughout this section, let $R$ and $S$ be rings, $M$ be a finite left $R$-module, $N$ be a finite right $S$-module, and let $\Omega$ be a non-empty finite set with $|\Omega|=n$. For any $\alpha\in M^{\Omega}$, the \textit{rank support} of $\alpha$, denoted by $\sigma(\alpha)$, is defined as the left $R$-submodule of $M$ generated by $\{\alpha_i\mid i\in\Omega\}$, i.e.,
\begin{equation}\sigma(\alpha)=\sum_{i\in\Omega} R\alpha_i;\end{equation}
moreover, the \textit{rank weight} of $\alpha$ is defined as the quantity
\begin{equation}\rank_R(\sigma(\alpha)).\end{equation}
We are interested in three partitions of $M^{\Omega}$ induced by rank support and rank weight. More precisely, \textit{the partition of $M^{\Omega}$ induced by $\sigma$} is the partition $\Lambda_1$ of $M^{\Omega}$ defined as
\begin{equation}\alpha\sim_{\Lambda_1}\gamma\Longleftrightarrow\sigma(\alpha)=\sigma(\gamma);\end{equation}
\textit{the partition of $M^{\Omega}$ induced by isomorphic relation of $\sigma$} is the partition $\Lambda_2$ of $M^{\Omega}$ defined as
\begin{equation}\alpha\sim_{\Lambda_2}\gamma\Longleftrightarrow\text{$\sigma(\alpha)\cong\sigma(\gamma)$ as left $R$-modules};\end{equation}
\textit{the partition of $M^{\Omega}$ induced by rank weight of $\sigma$} is the partition $\Lambda_3$ of $M^{\Omega}$ defined as
\begin{equation}\alpha\sim_{\Lambda_3}\gamma\Longleftrightarrow\rank_R(\sigma(\alpha))=\rank_R(\sigma(\gamma)).\end{equation}
It can be readily checked that $\Lambda_1$ is finer than $\Lambda_2$, and $\Lambda_2$ is finer than $\Lambda_3$.

\setlength{\parindent}{2em}
\begin{example}
Let $\mathbb{F}$ be a finite field, $m,n\in\mathbb{Z}^{+}$, and set $R=\mathbb{F}$, $M=\mathbb{F}^{[m]}$, $\Omega=[1,n]$. We identify $\Mat_{m,n}(\mathbb{F})$ with $M^{n}$ via the column blocking of matrices. Then, for any $\alpha\in\Mat_{m,n}(\mathbb{F})$, $\sigma(\alpha)$ is equal to the $\mathbb{F}$-subspace of $M$ generated by all the columns of $\alpha$, and the rank weight of $\alpha$ is exactly the rank of the matrix $\alpha$. Moreover, $\Lambda_1$ recovers the column space partition of $\Mat_{m,n}(\mathbb{F})$, and both $\Lambda_2$ and $\Lambda_3$ recover the rank partition of $\Mat_{m,n}(\mathbb{F})$ (see [9, Definitions 2.1 and 2.2]).
\end{example}

We note that rank support and rank weight can be defined in a parallel way for the right $S$-module $N$. More precisely, for any $\beta\in N^{\Omega}$, the \textit{rank support} of $\beta$, denoted by $\tau(\beta)$, is defined as the right $S$-submodule of $N$ generated by $\{\beta_i\mid i\in\Omega\}$, i.e.,
\begin{equation}\tau(\beta)=\sum_{i\in\Omega}\beta_iS,\end{equation}
and the \textit{rank weight} of $\beta$ is defined as the quantity
\begin{equation}\rank_S(\tau(\beta)).\end{equation}
Moreover, the partitions of $N^{\Omega}$ induced by $\tau$, isomorphic relation of $\tau$ and rank weight of $\tau$ are respectively the partitions $\Psi_1$, $\Psi_2$ and $\Psi_3$ of $N^{\Omega}$ defined as
\begin{equation}\beta\sim_{\Psi_1}\theta\Longleftrightarrow\tau(\beta)=\tau(\theta),\end{equation}
\begin{equation}\beta\sim_{\Psi_2}\theta\Longleftrightarrow\text{$\tau(\beta)\cong\tau(\theta)$ as right $S$-modules},\end{equation}
\begin{equation}\beta\sim_{\Psi_3}\theta\Longleftrightarrow\rank_S(\tau(\beta))=\rank_S(\tau(\theta)).\end{equation}

In this section, we focus on partitions induced by rank support and isomorphic relation of rank support. More precisely, we will show that with some additional assumptions, $(\Lambda_1,\Psi_1)$ and $(\Lambda_2,\Psi_2)$ are mutually dual.

Throughout the rest of this section, let $U$ be an $(R,S)$-bimodule, $\varpi:M\times N\longrightarrow U$ be a non-degenerate $(R,S)$-bilinear map, and define $\langle~,~\rangle:M^{\Omega}\times N^{\Omega}\longrightarrow U$ as
\begin{equation}\langle\alpha,\beta\rangle=\sum_{i\in\Omega}\varpi(\alpha_i,\beta_i).\end{equation}
Moreover, let $\chi\in\Hom((U,+),\mathbb{C}^{\ast})$ such that the only left $R$-submodule or right $S$-submodule of $U$ contained in $\ker(\chi)$ is $\{0\}$, and define $f:M^{\Omega}\times N^{\Omega}\longrightarrow \mathbb{C}^{\ast}$ as
\begin{equation}f(\alpha,\beta)=\chi(\langle\alpha,\beta\rangle).\end{equation}
For notational simplicity, for any $A\subseteq M$ and $B\subseteq N$, define
\begin{equation}A^{\ddagger}=\{v\in N\mid\forall~u\in A:\varpi(u,v)=0\},\end{equation}
\begin{equation}{^{\ddagger}B}=\{u\in M\mid \forall~v\in B:\varpi(u,v)=0\}.\end{equation}

We begin by computing the Krawtchouk matrices.

\setlength{\parindent}{0em}
\begin{proposition}
{\bf{(1)}}\,\,Let $\Gamma=\{B\mid\text{$B\leqslant_{S}N$}\}$, and let $\mu$ be the M\"{o}bius function of $(\Gamma,\subseteq)$. Consider $\alpha\in M^{\Omega}$ with $A\triangleq\sigma(\alpha)$. Then, for any $V\leqslant_{S}N$, it holds that
$$\sum_{(\beta\in N^{\Omega},\tau(\beta)=V)}f(\alpha,\beta)=\sum_{B\leqslant_{S}(V\cap{A^{\ddagger}})}|B|^{n}\mu(B,V).$$
{\bf{(2)}}\,\,Let $\Delta=\{A\mid\text{$A\leqslant_{R}M$}\}$, and let $\lambda$ be the M\"{o}bius function of $(\Delta,\subseteq)$. Consider $\theta\in N^{\Omega}$ with $B\triangleq\tau(\theta)$. Then, for any $L\leqslant_{R}M$, it holds that
$$\sum_{(\gamma\in M^{\Omega},\sigma(\gamma)=L)}f(\gamma,\theta)=\sum_{K\leqslant_{R}(L\cap{^{\ddagger}B})}|K|^{n}\lambda(K,L).$$
\end{proposition}

\begin{proof}
We only prove (1) and the proof of (2) is similar. Define $\varphi,\psi:\Gamma\longrightarrow\mathbb{C}$ as
$$\varphi(V)=\sum_{(\beta\in N^{\Omega},\tau(\beta)=V)}f(\alpha,\beta),$$
$$\psi(V)=\sum_{(\beta\in N^{\Omega},\tau(\beta)\subseteq V)}f(\alpha,\beta).$$
It follows from the definitions of $\varphi,\psi$ that for any $V\leqslant_{S}N$, it holds that
$$\psi(V)=\sum_{(U\in\Gamma,U\subseteq V)}\varphi(U).$$
From the M\"{o}bius inversion formula of the partially ordered set $(\Gamma,\subseteq)$ (see [19, Proposition 3.7.1]), we deduce that for any $V\leqslant_{S}N$, it holds that
\begin{equation}\varphi(V)=\sum_{(U\in\Gamma,U\subseteq V)}\psi(U)\mu(U,V).\end{equation}
Now we compute $\psi(U)$ for an arbitrary $U\leqslant_{S}N$. Noticing that for any $\beta\in N^{\Omega}$, it holds that $\tau(\beta)\subseteq U\Longleftrightarrow\beta\in U^{\Omega}$, an application of the orthogonal relation implies that
$$\psi(U)=\sum_{\beta\in U^{\Omega}}f(\alpha,\beta)=\begin{cases}
|U^{\Omega}|,&\forall~\beta\in U^{\Omega}:f(\alpha,\beta)=1;\\
0,&\exists~\beta\in U^{\Omega}~s.t.~f(\alpha,\beta)\neq1.
\end{cases}
$$
Noticing that $A=\sigma(\alpha)$, from Lemma 2.2 and some straightforward verification, we deduce that
$$(\forall~\beta\in U^{\Omega}:f(\alpha,\beta)=1)\Longleftrightarrow(\forall~\beta\in U^{\Omega}:\langle\alpha,\beta\rangle=0)\Longleftrightarrow\alpha\in({^{\ddagger}U})^{\Omega}\Longleftrightarrow A\subseteq{^{\ddagger}U}\Longleftrightarrow U\subseteq A^{\ddagger},$$
which further implies that
\begin{equation}\psi(U)=\begin{cases}
|U|^{n},&U\subseteq A^{\ddagger};\\
0,&U\nsubseteq A^{\ddagger}.
\end{cases}
\end{equation}
By (3.15) and (3.16), for any $V\leqslant_{S}N$, we have
$$\varphi(V)=\sum_{(U\in\Gamma,U\subseteq V,U\subseteq A^{\ddagger})}|U|^{n}\mu(U,V)=\sum_{U\leqslant_{S}(V\cap A^{\ddagger})}|U|^{n}\mu(U,V),$$
as desired.
\end{proof}

\setlength{\parindent}{2em}
Proposition 3.1 and (2) of Lemma 2.1 immediately implies that the partitions $\Lambda_1$ and $\Psi_1$ are mutually dual and hence reflexive. More precisely, we have the following theorem.

\setlength{\parindent}{0em}
\begin{theorem}
$(\Lambda_1,\Psi_1)$ is mutually dual with respect to $f$. Consequently, we have $|\Lambda_1|=|\Psi_1|$, and both $\Lambda_1$ and $\Psi_1$ are reflexive.
\end{theorem}

\setlength{\parindent}{2em}
\begin{example}
Let $\mathbb{F}$ be a finite field with $|\mathbb{F}|=q$, and let $m,n\in\mathbb{Z}^{+}$. Suppose that $R=S=U=\mathbb{F}$, $M=N=\mathbb{F}^{[m]}$, $\Omega=[1,n]$, and $\omega:\mathbb{F}^{[m]}\times \mathbb{F}^{[m]}\longrightarrow \mathbb{F}$ is the standard inner product on $\mathbb{F}^{[m]}$. Identifying $\Mat_{m,n}(\mathbb{F})$ with $M^{n}$ as in Example 3.1, $\langle~,~\rangle:\Mat_{m,n}(\mathbb{F})\times\Mat_{m,n}(\mathbb{F})\longrightarrow\mathbb{F}$ coincides with the trace inner product. By [9, Equation (3.1)], for any $U,V\leqslant_{\mathbb{F}}\mathbb{F}^{[m]}$ with $U\subseteq V$, we have
$$\mu(U,V)=(-1)^{b-a}q^{\binom{b-a}{2}},$$
where $a=\dim_{\mathbb{F}}(U)$, $b=\dim_{\mathbb{F}}(V)$. Hence for any $\alpha\in\Mat_{m,n}(\mathbb{F})$ with $A\triangleq\sigma(\alpha)$ and $V\leqslant_{\mathbb{F}}\mathbb{F}^{[m]}$ with $v\triangleq\dim_{\mathbb{F}}(V)$, an application of Proposition 3.1 implies that
$$\sum_{(\beta\in\Mat_{m,n}(\mathbb{F}),\sigma(\beta)=V)}f(\alpha,\beta)=\sum_{t=0}^{v}(-1)^{v-t}q^{tn+\binom{v-t}{2}}\bin_{q}(\dim_{\mathbb{F}}(V\cap A^{\ddagger}),t),$$
where for any $(b,a)\in\mathbb{N}^{2}$, $\bin_{q}(b,a)$ denotes the number of all the $a$-dimensional $\mathbb{F}$-subspaces of $\mathbb{F}^{b}$. This recovers [9, Theorem 3.3], which further leads to the MacWilliams identities for column support distributions of matrix codes (see [9, Corollary 3.4]). Moreover, Theorem 3.1 implies that $(\Lambda_1,\Lambda_1)$ is mutually dual with respect to $f$, which recovers [9, Corollary 2.6].
\end{example}

\setlength{\parindent}{2em}
Next, we turn to partitions induced by isomorphic relation of rank support, and we begin with the following proposition.

\setlength{\parindent}{0em}
\begin{proposition}
Let $\alpha,\gamma\in M^{\Omega}$ with $A\triangleq\sigma(\alpha)$, $E\triangleq\sigma(\gamma)$. Suppose that $\psi\in\Aut_{S}(N)$ and $E^{\ddagger}=\psi[A^{\ddagger}]$. Then, for any $V\leqslant_{S}N$, it holds that
$$\sum_{(\theta\in N^{\Omega},\tau(\theta)=\psi[V])}f(\gamma,\theta)=\sum_{(\beta\in N^{\Omega},\tau(\beta)=V)}f(\alpha,\beta).$$
\end{proposition}

\begin{proof}
Let $\Gamma=\{B\mid\text{$B\leqslant_{S}N$}\}$, and let $\mu$ be the M\"{o}bius function of $(\Gamma,\subseteq)$. Consider $V\leqslant_{S}N$. Noticing that $\psi[V]\cap{E^{\ddagger}}=\psi[V]\cap\psi[A^{\ddagger}]=\psi[V\cap{A^{\ddagger}}]$, by (1) of Proposition 3.1, we have
\begin{eqnarray*}
\begin{split}
\sum_{(\theta\in N^{\Omega},\tau(\theta)=\psi[V])}f(\gamma,\theta)&=\sum_{D\leqslant_{S}(\psi[V]\cap{E^{\ddagger}})}|D|^{n}\mu(D,\psi[V])\\
&=\sum_{D\leqslant_{S}\psi[V\cap{A^{\ddagger}}]}|D|^{n}\mu(D,\psi[V])\\
&=\sum_{B\leqslant_{S}(V\cap{A^{\ddagger}})}|\psi[B]|^{n}\mu(\psi[B],\psi[V])\\
&=\sum_{B\leqslant_{S}(V\cap{A^{\ddagger}})}|B|^{n}\mu(B,V)\\
&=\sum_{(\beta\in N^{\Omega},\tau(\beta)=V)}f(\alpha,\beta),
\end{split}
\end{eqnarray*}
as desired.
\end{proof}

\setlength{\parindent}{2em}
Now we are ready to show that the partitions $\Lambda_2$ and $\Psi_2$ are mutually dual and hence reflexive, as detailed in the following theorem.

\setlength{\parindent}{0em}
\begin{theorem}
Suppose that $M$ and $N$ are pseudo-injective; moreover, for any simple left $R$-module $X$, $\Hom_{R}(X,U)$ is a simple right $S$-module, and for any simple right $S$-modules $Y$, $\Hom_{S}(Y,U)$ is a simple left $R$-module. Then, $(\Lambda_2,\Psi_2)$ is mutually dual with respect to $f$. Consequently, we have $|\Lambda_2|=|\Psi_2|$, and both $\Lambda_2$ and $\Psi_2$ are reflexive.
\end{theorem}

\begin{proof}
First, let $\alpha,\gamma\in M^{\Omega}$ with $\alpha\sim_{\Lambda_2}\gamma$, and let $A\triangleq\sigma(\alpha)$, $E\triangleq\sigma(\gamma)$. Since  $A\cong E$ as left $R$-modules, by Lemma 2.3, $A^{\ddagger}\cong E^{\ddagger}$ as right $S$-modules; moreover, since $N$ is pseudo-injective, we can choose $\psi\in\Aut_{S}(N)$ with $E^{\ddagger}=\psi[A^{\ddagger}]$. Now Consider an arbitrary $B\in\Psi_2$. Let $\xi\in B$, and let $W=\tau(\xi)$. From the definition of $\Psi_2$, we infer that $B=\{\beta\in N^{\Omega}\mid \text{$\tau(\beta)\cong W$ as right $S$-modules}\}$. Moreover, for any $V\leqslant_{S}N$, we have $\psi[V]\leqslant_{S}N$ and $\psi[V]\cong V$ as right $S$-modules. Hence from (1) of Proposition 3.2, we deduce that
\begin{eqnarray*}
\begin{split}
\sum_{\beta\in B}f(\alpha,\beta)&=\sum_{(V\leqslant_{S}N,V\cong W)}\sum_{(\beta\in N^{\Omega},\tau(\beta)=V)}f(\alpha,\beta)\\
&=\sum_{(V\leqslant_{S}N,V\cong W)}\sum_{(\theta\in N^{\Omega},\tau(\theta)=\psi[V])}f(\gamma,\theta)\\
&=\sum_{(V\leqslant_{S}N,V\cong W)}\sum_{(\theta\in N^{\Omega},\tau(\theta)=V)}f(\gamma,\theta)\\
&=\sum_{\theta\in B}f(\gamma,\theta).
\end{split}
\end{eqnarray*}
It then follows that $\alpha\sim_{\textbf{\textit{l}}(\Psi_2)}\gamma$, where $\textbf{\textit{l}}(\Psi_2)$ denotes the left dual partition of $\Psi_2$ with respect to $f$. Hence $\Lambda_2$ is finer than $\textbf{\textit{l}}(\Psi_2)$. Similarly, one can show that $\Psi_2$ is finer than $\textbf{\textit{r}}(\Lambda_2)$, the right dual partition of $\Lambda_2$ with respect to $f$. Therefore $(\Lambda_2,\Psi_2)$ is mutually dual with respect to $f$, and an application of (2) of Lemma 2.1 immediately establishes the desired result.
\end{proof}

\setlength{\parindent}{2em}
We end this section by noting that all the results established in this section can be applied to free modules over finite Frobenius rings. More precisely, suppose that $R$ is a finite Frobenius ring, $S=U=R$, $m\in\mathbb{Z}^{+}$, $M=N=R^{[m]}$, $\varpi$ is the standard inner product on $R^{[m]}$, and $\chi$ is a generating character of $R$. Since a Frobenius ring is quasi-Frobenius (see [20, Section 1]), all the assumptions in Theorem 3.2 are satisfied (see [4, Section 58] for more details), and hence Propositions 3.1, 3.2 and Theorems 3.1, 3.2 can be readily applied.

\section{Non-reflexive partitions induced by rank weight}
\setlength{\parindent}{2em}
Now we turn to partitions of modules over finite chain rings induced by rank weight. We will show that unlike partitions induced by rank support, these partitions are non-reflexive except for the familiar case for finite fields.

Throughout this section, let $R$ be a finite chain ring, and fix $\pi\in R$ such that $R\pi$ is the unique maximal left ideal of $R$. Let $q\triangleq|R/R\pi|$, and let $s\in\mathbb{Z}^{+}$ such that $|R|=q^{s}$. Fix $m,n\in\mathbb{Z}^{+}$ with $m\geqslant2$, $n\geqslant2$. Let $M$ be a free finite left $R$-module with $\rank_R(M)=m$, $N$ be a free finite right $R$-module with $\rank_R(N)=m$, and let $\Omega$ be a non-empty finite set with $|\Omega|=n$. For any $\alpha\in M^{\Omega}$, $\beta\in N^{\Omega}$, let $\sigma(\alpha)=\sum_{i\in\Omega}R\alpha_i$, $\tau(\beta)=\sum_{i\in\Omega}\beta_iR$. As in Section 3, let $\varpi:M\times N\longrightarrow R$ be a non-degenerated $(R,R)$-bilinear map, $\chi$ be a generating character of $R$, and define the non-degenerate pairing $f:M^{\Omega}\times N^{\Omega}\longrightarrow \mathbb{C}^{\ast}$ as
\begin{equation}f(\alpha,\beta)=\chi\left(\sum_{i\in\Omega}\varpi(\alpha_i,\beta_i)\right).\end{equation}
Following Section 3, we let $\Lambda_2$ and $\Lambda_3$ denote the partitions of $M^{\Omega}$ induced by isomorphic relation of $\sigma$ and rank weight of $\sigma$, respectively, and let $\Psi_2$ and $\Psi_3$ denote the partitions of $N^{\Omega}$ induced by isomorphic relation of $\tau$ and rank weight of $\tau$, respectively; moreover, for $i\in\{2,3\}$, let $\textbf{\textit{l}}(\Psi_i)$ denote the left dual partition of $\Psi_i$ with respect to $f$.

The following theorem is the main result of this section, in which we characterize the dual partition and non-reflexivity of $\Psi_3$.

\setlength{\parindent}{0em}
\begin{theorem}
We have $\Lambda_2=\textbf{\textit{l}}(\Psi_3)=\textbf{\textit{l}}(\Psi_2)$. Consequently, if $R$ is not a field, then $\Psi_3$ is non-reflexive, and it holds that $\Lambda_2\neq \Lambda_3$, $\Psi_2\neq\Psi_3$.
\end{theorem}

\setlength{\parindent}{2em}
Theorem 4.1 will follow from the following more explicit characterization.

\setlength{\parindent}{0em}
\begin{theorem}
Let $\alpha,\gamma\in M^{\Omega}$. Then, the following three statements are equivalent to each other:

{\bf{(1)}}\,\,$\sigma(\alpha)\cong\sigma(\gamma)$ as left $R$-modules;

{\bf{(2)}}\,\,For any $b\in\mathbb{N}$, it holds that
\begin{eqnarray*}
\begin{split}
\sum_{(\beta\in N^{\Omega},\rank_R(\tau(\beta))=b)}f(\alpha,\beta)=\sum_{(\beta\in N^{\Omega},\rank_R(\tau(\beta))=b)}f(\gamma,\beta);
\end{split}
\end{eqnarray*}

{\bf{(3)}}\,\,It holds that
\begin{eqnarray*}
\begin{split}
\sum_{(\beta\in N^{\Omega},\rank_R(\tau(\beta))=1)}f(\alpha,\beta)=\sum_{(\beta\in N^{\Omega},\rank_R(\tau(\beta))=1)}f(\gamma,\beta).
\end{split}
\end{eqnarray*}
\end{theorem}

\setlength{\parindent}{2em}
Now we prove Theorem 4.2. Our proof is based on two lemmas. Throughout the rest of this section, as in (3.13), for any $A\subseteq M$, set $A^{\ddagger}=\{v\in N\mid\forall~u\in A:\varpi(u,v)=0\}$.

\setlength{\parindent}{2em}
\begin{lemma}
Let $\alpha\in M^{\Omega}$ with $A\triangleq\sigma(\alpha)$, $V\leqslant_{R}N$ with $\rank_R(V)=1$, and let $t\in\mathbb{Z}^{+}$ with $|V|=q^{t}$. Then, it holds that
\begin{eqnarray*}
\begin{split}
\sum_{(\beta\in N^{\Omega},\tau(\beta)=V)}f(\alpha,\beta)=\begin{cases}
0,&V\pi\nsubseteq A^{\ddagger};\\
-q^{(t-1)n},&V\pi\subseteq A^{\ddagger},V\nsubseteq A^{\ddagger};\\
q^{tn}-q^{(t-1)n},&V\subseteq A^{\ddagger}.
\end{cases}
\end{split}
\end{eqnarray*}
\end{lemma}

\begin{proof}
Let $\Gamma=\{B\mid\text{$B\leqslant_{R}N$}\}$, and let $\mu$ be the M\"{o}bius function of $(\Gamma,\subseteq)$. By Proposition 3.1, we have
\begin{equation}w\triangleq\sum_{(\beta\in N^{\Omega},\tau(\beta)=V)}f(\alpha,\beta)=\sum_{B\leqslant_{R}(V\cap{A^{\ddagger}})}|B|^{n}\mu(B,V).\end{equation}
Let $T=\{B\mid B\leqslant_{R}V\}=\{B\mid B\in\Gamma,B\subseteq V\}$. Since $R$ is a finite chain ring, $\rank_R(V)=1$, $|V|=q^{t}$, we infer that $(T,\subseteq)$ is a chain, $|V\pi|=q^{t-1}$, and $V\pi$ is the unique maximal element of $(T-\{V\},\subseteq)$. Consequently, we have $\mu(V,V)=1$, $\mu(V\pi,V)=-1$, and for any $B\leqslant_{R}V$ with $B\subsetneqq V\pi$, it holds that $\mu(B,V)=0$.

Now we discuss in three cases. First, suppose that $V\pi\nsubseteq A^{\ddagger}$. Then, we have $V\cap{A^{\ddagger}}\subsetneqq V\pi$, which implies that $\mu(B,V)=0$ for all $B\leqslant_{R}(V\cap{A^{\ddagger}})$. It then follows from (4.2) that $w=0$, as desired. Second, suppose that $V\pi\subseteq A^{\ddagger}$ and $V\nsubseteq A^{\ddagger}$. Then, we have $V\cap{A^{\ddagger}}=V\pi$. It then follows from (4.2) that
$$w=\sum_{B\leqslant_{R}V\pi}|B|^{n}\mu(B,V)=|V\pi|^{n}\mu(V\pi,V)=-q^{(t-1)n},$$
as desired. Third, suppose that $V\subseteq A^{\ddagger}$. By (4.2), we have
$$w=\sum_{B\leqslant_{R}V}|B|^{n}\mu(B,V)=|V\pi|^{n}\mu(V\pi,V)+|V|^{n}\mu(V,V)=-q^{(t-1)n}+q^{tn},$$
as desired.
\end{proof}

\setlength{\parindent}{0em}
\begin{lemma}
Let $\alpha\in M^{\Omega}$ with $A\triangleq\sigma(\alpha)$; moreover, for any $i\in[0,s]$, let $c_{(i)}\in\mathbb{N}$ such that
$$|\{v\in A^{\ddagger}\mid v\pi^{i}=0\}|=q^{c_{(i)}}.$$
Then, it holds that
\begin{eqnarray*}
\begin{split}
(q-1)\sum_{(\beta\in N^{\Omega},\rank_{R}(\tau(\beta))=1)}f(\alpha,\beta)=&\left(\sum_{t=1}^{s}q^{c_{(t)}+t(n-1)+1}\right)+\left(\sum_{t=1}^{s-1}q^{c_{(t-1)}+t(n-1)+m}\right)\\
&-\left(\sum_{t=1}^{s}q^{c_{(t-1)}+t(n-1)+1}\right)-\left(\sum_{t=1}^{s-1}q^{c_{(t)}+t(n-1)+m}\right)-(q^{m}-1).
\end{split}
\end{eqnarray*}
\end{lemma}

\begin{proof}
Let $D=\{z\in N \mid z\pi \in A^{\ddagger}\}$. Noticing that for any $V\leqslant_{R}N$, we have $V\pi\subseteq A^{\ddagger}\Longleftrightarrow V\subseteq D$. We first consider an arbitrary $t\in\{1,\dots,s\}$. By Lemma 4.1 and (3), (4) of Lemma 2.4, we have
\begin{eqnarray*}
\begin{split}
w_t&\triangleq\sum_{(V\leqslant_{R}N,\rank_{R}(V)=1,|V|=q^{t})}\sum_{(\beta\in N^{\Omega},\tau(\beta)=V)}f(\alpha,\beta)\\
&=(q^{tn}-q^{(t-1)n})|\{V\leqslant_{R}A^{\ddagger}\mid\rank_{R}(V)=1,|V|=q^{t}\}|\\
&\,\,\,\,\,\,\,\,-q^{(t-1)n}|\{V\leqslant_{R}D\mid\rank_{R}(V)=1,|V|=q^{t},V\nsubseteq A^{\ddagger}\}|\\
&=q^{tn}|\{V\leqslant_{R}A^{\ddagger}\mid\rank_{R}(V)=1,|V|=q^{t}\}|-q^{(t-1)n}|\{V\leqslant_{R}D\mid\rank_{R}(V)=1,|V|=q^{t}\}|\\
&=q^{tn}\cdot\frac{q^{c_{(t)}}-q^{c_{(t-1)}}}{q^{t}-q^{t-1}}-q^{(t-1)n}\cdot\frac{|\{y\in D\mid y\pi^{t}=0\}|-|\{y\in D\mid y\pi^{t-1}=0\}|}{q^{t}-q^{t-1}}\\
&=\begin{cases}
q^{tn}\cdot\frac{q^{c_{(t)}}-q^{c_{(t-1)}}}{q^{t}-q^{t-1}}-q^{(t-1)n}\cdot\frac{q^{c_{(t-1)}+m}-q^{c_{(t-2)}+m}}{q^{t}-q^{t-1}},&t\geqslant2;\\
q^{tn}\cdot\frac{q^{c_{(t)}}-q^{c_{(t-1)}}}{q^{t}-q^{t-1}}-\frac{q^{m}-1}{q-1},&t=1,
\end{cases}
\end{split}
\end{eqnarray*}
which further implies that
\begin{eqnarray*}
\begin{split}
(q-1)w_t=\begin{cases}
q^{t(n-1)+c_{(t)}+1}-q^{t(n-1)+c_{(t-1)}+1}-q^{(t-1)(n-1)+c_{(t-1)}+m}+q^{(t-1)(n-1)+c_{(t-2)}+m},&t\geqslant2;\\
q^{t(n-1)+c_{(t)}+1}-q^{t(n-1)+c_{(t-1)}+1}-(q^{m}-1),&t=1.
\end{cases}
\end{split}
\end{eqnarray*}
The above discussion implies that
\begin{eqnarray*}
\begin{split}
h&\triangleq(q-1)\sum_{(\beta\in N^{\Omega},\rank_{R}(\tau(\beta))=1)}f(\alpha,\beta)\\
&=\sum_{t=1}^{s}(q-1)w_t\\
&=\left(\sum_{t=1}^{s}q^{t(n-1)+c_{(t)}+1}\right)-\left(\sum_{t=1}^{s}q^{t(n-1)+c_{(t-1)}+1}\right)-\left(\sum_{t=2}^{s}q^{(t-1)(n-1)+c_{(t-1)}+m}\right)\\
&\,\,\,\,\,\,\,\,+\left(\sum_{t=2}^{s}q^{(t-1)(n-1)+c_{(t-2)}+m}\right)-(q^{m}-1),
\end{split}
\end{eqnarray*}
which further establishes the desired result.
\end{proof}

\setlength{\parindent}{2em}
We are ready to establish Theorem 4.2.

\begin{proof}[Proof of Theorem 4.2]
${\bf{(1)}}\Longrightarrow{\bf{(2)}}$\,\,Noticing that all the assumptions of Theorem 3.2 are satisfied with $R=S=U$, we have $\alpha\sim_{\textbf{\textit{l}}(\Psi_2)}\gamma$. Since $\Psi_2$ is finer than $\Psi_3$, $\textbf{\textit{l}}(\Psi_2)$ is finer than $\textbf{\textit{l}}(\Psi_3)$, which further implies that $\alpha\sim_{\textbf{\textit{l}}(\Psi_3)}\gamma$, establishing (2), as desired.

${\bf{(2)}}\Longrightarrow{\bf{(3)}}$\,\,This is trivial.

${\bf{(3)}}\Longrightarrow{\bf{(1)}}$\,\,Let $A\triangleq\sigma(\alpha)$, $E\triangleq\sigma(\gamma)$; moreover, for any $i\in[0,s]$, let $c_{(i)},e_{(i)}\in\mathbb{N}$ such that
$$|\{v\in A^{\ddagger}\mid v\pi^{i}=0\}|=q^{c_{(i)}},~|\{v\in E^{\ddagger}\mid v\pi^{i}=0\}|=q^{e_{(i)}}.$$
Now we show by induction that $c_{(i)}=e_{(i)}$ for all $i\in[0,s]$. First of all, one can readily check that $c_{(0)}=e_{(0)}=0$. Now let $p\in[1,s]$ such that $c_{(i)}=e_{(i)}$ for all $i\in[0,p-1]$, and we will show that $c_{(p)}=e_{(p)}$. Indeed, from Lemma 4.2, we deduce that
\begin{equation}
\begin{aligned}
&\left(\sum_{t=p}^{s}q^{c_{(t)}+t(n-1)+1}\right)+\left(\sum_{t\in[p+1,s-1]}q^{c_{(t-1)}+t(n-1)+m}\right)-\left(\sum_{t=p+1}^{s}q^{c_{(t-1)}+t(n-1)+1}\right)-\left(\sum_{t=p}^{s-1}q^{c_{(t)}+t(n-1)+m}\right)\\
&=\left(\sum_{t=p}^{s}q^{e_{(t)}+t(n-1)+1}\right)+\left(\sum_{t\in[p+1,s-1]}q^{e_{(t-1)}+t(n-1)+m}\right)-\left(\sum_{t=p+1}^{s}q^{e_{(t-1)}+t(n-1)+1}\right)-\left(\sum_{t=p}^{s-1}q^{e_{(t)}+t(n-1)+m}\right).
\end{aligned}
\end{equation}
Let $w\in\mathbb{Z}$ denote both sides of (4.3). We claim that
\begin{equation}\text{$q^{c_{(p)}+p(n-1)+1}\mid w$ and $q^{c_{(p)}+p(n-1)+2}\nmid w$}.\end{equation}
Indeed, for any $t\in[p+1,s]$, from $c_{(t-1)}\geqslant c_{(p)}$ and $n\geqslant2$, we deduce that
$$c_{(t-1)}+t(n-1)+1\geqslant c_{(p)}+p(n-1)+2,$$
which, together with $c_{(t)}\geqslant c_{(t-1)}$ and $m\geqslant2$, further implies that
$$c_{(t)}+t(n-1)+1\geqslant c_{(p)}+p(n-1)+2,$$
$$c_{(t-1)}+t(n-1)+m\geqslant c_{(p)}+p(n-1)+2.$$
Moreover, for any $t\in[p,s-1]$, from $c_{(t)}\geqslant c_{(p)}$ and $m\geqslant2$, we deudce that
$$c_{(t)}+t(n-1)+m\geqslant c_{(p)}+p(n-1)+2.$$
The above discussion establishes (4.4), as desired. Moreover, a parallel discussion implies that
\begin{equation}\text{$q^{e_{(p)}+p(n-1)+1}\mid w$ and $q^{e_{(p)}+p(n-1)+2}\nmid w$}.\end{equation}
From (4.4) and (4.5), we deduce that $c_{(p)}+p(n-1)+1=e_{(p)}+p(n-1)+1$, which implies that $c_{(p)}=e_{(p)}$, as desired. Now an application of (2) of Lemma 2.4 implies that $A^{\ddagger}\cong E^{\ddagger}$ as right $R$-modules, which, together with Lemma 2.3, implies that $A\cong E$ as left $R$-modules, as desired.
\end{proof}

\setlength{\parindent}{2em}
Finally, we prove Theorem 4.1.

\begin{proof}[Proof of Theorem 4.1]
First, it follows from Theorems 4.2 and 3.2 that $\Lambda_2=\textbf{\textit{l}}(\Psi_3)$ and $\Lambda_2=\textbf{\textit{l}}(\Psi_2)$, as desired. Next, suppose that $R$ is not a field. Since $M$ is free, we can choose $z\in M$ such that $az\neq0$ for all $a\in R-\{0\}$. Let $\alpha\in M^{\Omega}$ such that $\alpha_i=z$ for all $i\in\Omega$, and let $\gamma\in M^{\Omega}$ such that $\gamma_i=\pi z$ for all $i\in\Omega$. Then, we have $\sigma(\alpha)=Rz$, $\sigma(\gamma)=R(\pi z)=(R\pi)z$. Since $R$ is not a field, $\pi\neq0$, which implies that $\rk_R(\sigma(\alpha))=\rk_R(\sigma(\gamma))=1$, and hence $\alpha\sim_{\Lambda_3}\gamma$; moreover, noticing that $\sigma(\alpha)\not\cong\sigma(\gamma)$ as left $R$-modules, we have $\alpha\not\sim_{\Lambda_2}\gamma$, which further implies that $\Lambda_2\neq\Lambda_3$, as desired. A similar discussion yields that $\Psi_2\neq\Psi_3$. Finally, noticing that $\Lambda_2$ is finer than $\Lambda_3$, we have $|\Lambda_2|>|\Lambda_3|=|\Psi_3|=\min\{m,n\}+1$. It then follows from Lemma 2.1 that $\Psi_3$ is non-reflexive, as desired.
\end{proof}

\section{Rank partitions of matrix spaces over finite chain rings}

\setlength{\parindent}{2em}
In this section, we study the dual partitions and reflexivity of rank partitions of matrix spaces over finite chain rings. To this end, let $R$ be a finite chain ring, and let $m,n\in\mathbb{Z}^{+}$ with $m\geqslant2$, $n\geqslant2$. Two matrices $\alpha,\gamma\in\Mat_{m,n}(R)$ are said to be \textit{equivalent} if there exists two invertible matrices $\eta_1\in\Mat_{m,m}(R)$ and $\eta_2\in\Mat_{n,n}(R)$ such that $\gamma=\eta_1\alpha\eta_2$. Similar to matrices over fields, any matrix $\alpha\in\Mat_{m,n}(R)$ is equivalent to a diagonal matrix $\beta\in\Mat_{m,n}(R)$ (i.e., $\beta_{i,j}=0$ for all $(i,j)\in[1,m]\times[1,n]$ with $i\neq j$), and such a $\beta$ is referred to as a Smith normal form of $\alpha$ (see [14, Section II.A]).

For any $\alpha\in\Mat_{m,n}(R)$, let $\col(\alpha)$ denote the right $R$-submodule of $R^{[m]}$ generated by all the columns of $\alpha$, and let $\row(\alpha)$ denote the left $R$-submodule of $R^{n}$ generated by all the rows of $\alpha$; moreover, following [14, Definition 3.3] and [2, Definition 2.5], the rank of $\alpha$, denoted by $\rk(\alpha)$, is defined as
\begin{equation}\rk(\alpha)=\rank_{R}(\col(\alpha))=\rank_{R}(\row(\alpha)),\end{equation}
where the second equality can be verified by using the Smith normal form (cf. [14, Corollary 3.5]). Similar to matrices over fields, $d:\Mat_{m,n}(R)\times \Mat_{m,n}(R)\longrightarrow\mathbb{N}$ defined as
\begin{equation}d(\alpha,\gamma)=\rk(\gamma-\alpha)\end{equation}
induces a metric on $\Mat_{m,n}(R)$, which is henceforth referred to as the \textit{rank metric} on $\Mat_{m,n}(R)$. Finally, the \textit{rank partition} $\Phi$ of $\Mat_{m,n}(R)$ is defined as
\begin{equation}\alpha\sim_{\Phi}\gamma\Longleftrightarrow\rk(\alpha)=\rk(\gamma).\end{equation}

When $R$ is a field, it is well-known that $\Phi$ is reflexive, and $(\Phi,\Phi)$ is mutually dual with respect to a suitable pairing (see [9, Corollary 2.6]). In terms of association schemes, this implies that the rank metric $d$ induces an association scheme, as has been shown in [6, Section 2]. MacWilliams identities for rank weight distributions of rank metric codes were first established by Delsarte in [6, Theorem 3.3], and have been re-established various times in recent years (see, e.g., \cite{8,11,17,18}).

Now we show that when $R$ is not a field, then opposite to the field case, $\Phi$ is non-reflexive, and the rank metric $d$ does not induce an association scheme. This will settle a question proposed in [2, Section 4], namely, whether rank metric on $\Mat_{m,n}(R)$ induces an association scheme. For further discussion, we let $\langle~,~\rangle:\Mat_{m,n}(R)^{2}\longrightarrow R$ denote the trace inner product, $\chi$ be a generating character of $R$, and define the non-degenerate pairing $f:\Mat_{m,n}(R)^{2}\longrightarrow\mathbb{C}^{*}$ as
\begin{equation}f(\alpha,\beta)=\chi(\langle\alpha,\beta\rangle).\end{equation}

We first use Theorems 4.1 and 4.2 to derive the dual partition and reflexivity of $\Phi$.

\setlength{\parindent}{0em}
\begin{theorem}
{\bf{(1)}}\,\,Let $\Delta$ be the partition of $\Mat_{m,n}(R)$ defined as $\alpha\sim_{\Delta}\gamma$ if and only if $\alpha^{T}$ and $\gamma^{T}$ are equivalent in $\Mat_{n,m}(R)$. Then, $\Delta$ is equal to the left dual partition of $\Phi$ with respect to $f$. Moreover, let $\alpha,\gamma\in \Mat_{m,n}(R)$. Then, $\alpha\sim_{\Delta}\gamma$ if and only if
$$\sum_{(\beta\in\Mat_{m,n}(R),\rk(\beta)=1)}f(\alpha,\beta)=\sum_{(\beta\in\Mat_{m,n}(R),\rk(\beta)=1)}f(\gamma,\beta).$$

{\bf{(2)}}\,\,$\Phi$ is reflexive if and only if $R$ is a field.
\end{theorem}

\begin{proof}
For an arbitrary $\alpha\in\Mat_{m,n}(R)$, we let $\alpha_i\in R^{[m]}$ denote the $i$-th column of $\alpha$ for all $i\in[1,n]$, and we identify $\alpha$ with $(\alpha_1,\dots,\alpha_n)\in(R^{[m]})^{n}$; moreover, set $\sigma(\alpha)=\sum_{i=1}^{n}R\alpha_i$, and set $\tau(\alpha)=\sum_{i=1}^{n}\alpha_iR=\col(\alpha)$. Let $\varpi:R^{[m]}\times R^{[m]}\longrightarrow R$ be the standard inner product on $R^{[m]}$. One can check that $\langle\alpha,\beta\rangle=\sum_{i=1}^{n}\varpi(\alpha_i,\beta_i)$ for all $\alpha,\beta\in\Mat_{m,n}(R)$. Since $\rk(\beta)=\rank_{R}(\tau(\beta))$ for all $\beta\in\Mat_{m,n}(R)$, $\Phi$ can be regarded as the partition of $(R^{[m]})^{n}$ induced by rank weight of $\tau$. Moreover, we note that for any $\alpha\in\Mat_{m,n}(R)$, the left $R$-submodule of $R^{m}$ generated by all the rows of $\alpha^{T}$, denoted by $\row(\alpha^{T})$, is equal to $\{\eta^{T}\mid\eta\in\sigma(\alpha)\}$.

{\bf{(1)}}\,\,It is straightforward to check that $\alpha\sim_{\Delta}\gamma$ if and only if $\row(\alpha^{T})\cong\row(\gamma^{T})$ as left $R$-modules, if and only if $\{\eta^{T}\mid\eta\in\sigma(\alpha)\}\cong\{\eta^{T}\mid\eta\in\sigma(\gamma)\}$ as left $R$-modules, if and only if $\sigma(\alpha)\cong\sigma(\gamma)$ as left $R$-modules. Hence Theorems 4.1 and 4.2 conclude the proof.

{\bf{(2)}}\,\,If $R$ is a field, then we have $\Delta=\Phi$, which, together with Lemma 2.1, implies that $\Phi$ is reflexive. If $R$ is not a field, then it follows from Theorem 4.1 that $\Phi$ is non-reflexive, as desired.
\end{proof}

\setlength{\parindent}{2em}
Finally, we consider the question that whether $d$ induces an association scheme. Following [5, Section 2.1], we recall that for a partition $\Upsilon$ of $\Mat_{m,n}(R)^{2}$, $(\Mat_{m,n}(R),\Upsilon)$ is said to be an association scheme if the following three conditions hold:

$(i)$\,\,$\{(\alpha,\alpha)\mid \alpha\in\Mat_{m,n}(R)\}\in\Upsilon$;

$(ii)$\,\,For any $W\in\Upsilon$, $W^{-1}=\{(\alpha,\beta)\mid(\beta,\alpha)\in W\}\in\Upsilon$;

$(iii)$\,\,For any $U,V,W\in\Upsilon$ and any $(\alpha,\beta),(\gamma,\theta)\in W$, it holds that
$$|\{\eta\in\Mat_{m,n}(R)\mid(\alpha,\eta)\in U,(\eta,\beta)\in V\}|=|\{\eta\in\Mat_{m,n}(R)\mid(\gamma,\eta)\in U,(\eta,\theta)\in V\}|.$$

\setlength{\parindent}{2em}
We are ready to state and prove the following theorem.

\setlength{\parindent}{0em}
\begin{theorem}
Let $\Upsilon$ be the partition of $\Mat_{m,n}(R)^{2}$ defined as $(\alpha,\gamma)\sim_{\Upsilon}(\beta,\theta)$ if and only if $d(\alpha,\gamma)=d(\beta,\theta)$. Then, $(\Mat_{m,n}(R),\Upsilon)$ is an association scheme if and only if $R$ is a field.
\end{theorem}

\begin{proof}
Noticing that $(\alpha,\gamma)\sim_{\Upsilon}(\beta,\theta)\Longleftrightarrow(\gamma-\alpha)\sim_{\Phi}(\theta-\beta)$ for all $(\alpha,\gamma),(\beta,\theta)\in\Mat_{m,n}(R)^{2}$, by [22, Theorem 1] (see Remark 2.1), $(\Mat_{m,n}(R),\Upsilon)$ is an association scheme if and only if $\Phi$ is reflexive. This, together with Theorem 5.1, immediately establishes the desired result.
\end{proof}

\section*{Appendix: Proof of Lemma 2.4}

\begin{proof}[Proof of Lemma 2.4]
{\bf{(1)}}\,\,First of all, (2.13) is a corollary of [12, Theorem 2.2] and [3, Section 2.1, Paragraph 3]. Next, consider $t\in[0,s]$. Noticing that $\{u\in R\mid\pi^{t}u=0\}=R\pi^{s-t}$, for any $i\in[1,m]$, we have $\{u\in R\pi^{s-\lambda_i}\mid\pi^{t}u=0\}=R\pi^{s-\lambda_i}\cap R\pi^{s-t}=R\pi^{s-\min\{\lambda_i,t\}}$, which implies that $|\{u\in R\pi^{s-\lambda_i}\mid\pi^{t}u=0\}|=q^{\min\{\lambda_i,t\}}$. It then follows from (2.13) that
$$\mbox{$|\{y\in M\mid \pi^{t}y=0\}|=|\{\alpha\in \prod_{i=1}^{m}R\pi^{s-\lambda_i}\mid \pi^{t}\alpha=0\}|=\prod_{i=1}^{m}q^{\min\{\lambda_i,t\}}=q^{\sum_{i=1}^{m}\min\{\lambda_i,t\}}$},$$
which further establishes the desired result.

{\bf{(2)}}\,\,By (1), there exists $m\in\mathbb{Z}^{+}$ and $(\lambda_1,\dots,\lambda_m),(\mu_1,\dots,\mu_m)\in[0,s]^{m}$ such that
\begin{equation}M\cong\prod_{i=1}^{m}R/R\pi^{\lambda_i},~P\cong\prod_{i=1}^{m}R/R\pi^{\mu_i}\end{equation}
as left $R$-modules. For any $t\in[0,s]$, set
$$a_t\triangleq\sum_{i=1}^{m}\min\{\lambda_i,t\},~b_t\triangleq\sum_{i=1}^{m}\min\{\mu_i,t\}.$$
From (1), we deduce that $a_t=b_t$ for all $t\in[1,s]$, and we also note that $a_0=b_0=0$. Via some straightforward computation, we deduce that for any $t\in[1,s]$, it holds that
$$|\{i\in [1,m]\mid\lambda_i\geqslant t\}|=a_t-a_{t-1}=b_t-b_{t-1}=|\{i\in [1,m]\mid\mu_i\geqslant t\}|.$$
This implies that $|\{i\in [1,m]\mid\lambda_i=t\}|=|\{i\in [1,m]\mid\mu_i=t\}|$ for all $t\in[1,s]$, which, together with (5.5), further implies that $M\cong P$ as left $R$-modules, as desired.

{\bf{(3)}}\,\,By [14, Proposition 3.2], there exists $(\lambda_1,\dots,\lambda_m)\in[0,s]^{m}$ and a left $R$-module isomorphism $\varphi:M\longrightarrow R^{m}$ such that $\varphi[A]=\prod_{i=1}^{m}R\pi^{s-\lambda_i}$ $^1$. \renewcommand{\thefootnote}{\fnsymbol{footnote}}\footnotetext{\hspace*{-6mm} \begin{tabular}{@{}r@{}p{16cm}@{}}
$^1$ &[14, Proposition 3.2] was established for commutative principal ideal rings; however, the result remains valid for \\&non-commutative finite chain rings as every matrix over a finite chain ring has a Smith normal form.
\end{tabular}} Define $(\mu_1,\dots,\mu_m)\in[0,s]^{m}$ as $\mu_i=\min\{\lambda_{i}+1,s\}$ for all $i\in[1,m]$. One can check that $\{u\in R\mid\pi u\in R\pi^{s-\lambda_i}\}=R\pi^{s-\mu_i}$ for all $i\in[1,m]$, which implies that $\varphi[B]=\prod_{i=1}^{m}R\pi^{s-\mu_i}$. Now consider an arbitrary $t\in[1,s]$. Let
\begin{equation}a=\sum_{i=1}^{m}\min\{\lambda_i,t-1\},~b=\sum_{i=1}^{m}\min\{\mu_i,t\}.\end{equation}
By (1), we have $|\{y\in A\mid \pi^{t-1}y=0\}|=q^{a}$, $|\{z\in B\mid \pi^{t}z=0\}|=q^{b}$. For any $i\in[1,m]$, from $s\geqslant t$, we deduce that $\min\{\mu_i,t\}=\min\{\lambda_{i}+1,t\}=\min\{\lambda_{i},t-1\}+1$. This, together with (5.6), implies that $b=a+m$, which further establishes the desired result.

{\bf{(4)}}\,\,This is a very special case of [12, Theorem 2.4], and we include a proof for completeness. Let $t\in[1,s]$, and let $Q=\{y\in M\mid \pi^{t}y=0,\pi^{t-1}y\neq0\}$. We infer that $\{Ry\mid y\in Q\}$ is exactly the set of all the left $R$-submodules of $M$ with rank $1$ and cardinality $q^{t}$. Next, we show that
\begin{equation}\forall~y\in Q:|\{z\in Q\mid Rz=Ry\}|=q^{t}-q^{t-1}.\end{equation}
To this end, consider $y\in Q$. Noticing that $R-J$ is exactly the set of all the multiplicative invertible elements of $R$, from [20, Proposition 5.1], we deduce that $\{z\in Q\mid Rz=Ry\}=\{uy\mid u\in R-J\}$. Now consider an arbitrary $u\in R-J$. Noticing that $\{a\in R\mid ay=0\}=R\pi^{t}$, we have $\{v\in R\mid vy=uy\}=u+R\pi^{t}$. Moreover, by $t\geqslant1$, we have $R\pi^{t}\subseteq J$ and hence $u+R\pi^{t}\subseteq R-J$, which implies that $|\{v\in R-J\mid vy=uy\}|=|R\pi^{t}|=q^{s-t}$. It then follows that
$$|\{uy\mid u\in R-J\}|=\frac{|R-J|}{q^{s-t}}=\frac{q^{s}-q^{s-1}}{q^{s-t}}=q^{t}-q^{t-1},$$
which further establishes (5.7), as desired. Finally, it follows from (5.7) that
$$|\{Ry\mid y\in Q\}|=\frac{|Q|}{q^{t}-q^{t-1}},$$
which immediately establishes the desired result.
\end{proof}


\begin{thebibliography}{9}

\bibitem{1}
F. W. Anderson, K. R. Fuller, \textit{Rings and Categories of Modules} (second edition), Springer, 1992.

\bibitem{2}
I. Blanco-Chac\'{o}n, A. F. Boix, M. Greferath, E. Hieta-Aho, MacWilliams duality for rank metric codes over finite chain rings, \textit{Finite Fields and Their Applications}, vol. 103 (2025), 102584.

\bibitem{3}
E. Byrne, Anna-Lena Horlemann, K. Khathuria, V. Weger, Density of free modules over finite chain rings, \textit{Linear Algebra and its Applications}, vol. 651 (2022), 1-25.

\bibitem{4}
C. W. Curtis, I. Reiner, \textit{Representation theory of finite groups and associative algebras}, Interscience, New York, 1962.


\bibitem{5}
P. Delsarte, An algebraic approach to the association schemes of coding theory, \textit{Philips Research Reports Supplements}, No. 10, 1973.

\bibitem{6}
P. Delsarte, Bilinear forms over a finite field, with applications to coding theory, \textit{Journal of Combinatorial Theory, Series A}. 25 (1978), 226-241.

\bibitem{7}
P. Delsarte, V. I. Levenshtein, Association schemes and coding theory, \textit{IEEE Transactions on Information Theory}, vol. 44, no. 6 (1998), 2477-2504.


\bibitem{8}
H. Gluesing-Luerssen, Fourier-Reflexive partitions and MacWilliams identities for additive codes, \textit{Designs, Codes and Cryptography}, vol. 75, no. 3 (2015), 543-563.

\bibitem{9}
H. Gluesing-Luerssen, A. Ravagnani, Partitions of matrix spaces with an application to $q$-rook polynomials, \textit{European Journal of Combinatorics}, vol. 89, no. 3 (2020), 103-120.


\bibitem{10}
E. Gorla, Rank-metric codes,  In W. Cary Huffman, Jon-Lark Kim, and Patrick Sol\'{e}, editors, Concise Encyclopedia of Coding Theory, 227-250. Chapman and Hall/CRC, 2021.

\bibitem{11}
E. Gorla and A. Ravagnani, Codes endowed with the rank metric, \textit{Network Coding and Subspace Designs}, 2018, 3-23.


\bibitem{12}
T. Honold, I. Landjev, Linear codes over finite chain rings, \textit{The Electronic Journal of Combinatorics}, vol. 7 (2000), R11.

\bibitem{13}
T. Honold, Characterization of finite Frobenius rings, \textit{Archiv der Mathematik}, vol. 76 (2001), 406-415.

\bibitem{14}
H. T. Kamche and C. Mouaha, Rank-metric codes over finite principal ideal rings and applications, \textit{IEEE Transactions on Information Theory}, vol. 65, no. 12 (2019), 7718-7735.


\bibitem{15} F. J. Macwilliams, A theorem on the distribution of weights in a systematic code, \textit{Bell Syst. Tech. J.}, vol. 42, no. 1 (1963), 79-94.


\bibitem{16}
P. Morandi, \textit{Field and Galois Theory (second edition)}, Springer-Verlag New York, USA (1996).

\bibitem{17}
A. Ravagnani, Rank-metric codes and their duality theory, \textit{Designs, Codes and Cryptography}, vol. 80, no. 1 (2016), 197-216.


\bibitem{18}
A. Ravagnani, Duality of codes supported on regular lattices, with an application to enumerative combinatorics, \textit{Designs, Codes and Cryptography}, vol. 86, no. 9 (2018), 2035-2063.


\bibitem{19}
P. Stanley, \textit{Enumerative combinatorics} (second edition), vol. 1, Cambridge University Press, Cambridge (2012).

\bibitem{20}
J. A. Wood, Duality for modules over finite rings and applications to coding theory, \textit{American Journal of Mathematics}, vol. 121, no. 3 (1999), 555-575.

\bibitem{21}
J. A. Wood, Foundations of linear codes defined over finite modules: the extension theorem and the MacWilliams identities, \textit{Codes Over Rings}, World Scientific Pub. Co. Inc., (2009), 124-190.

\bibitem{22}
V. A. Zinoviev, T. Ericson, Fourier-invariant pairs of partitions of finite abelian groups and association schemes, \textit{Problems of Information Transmission}, vol. 45, no. 3 (2009), 221-231.





























































































\end{thebibliography}
\end{document}